%% file: main.tex
\documentclass[12pt]{article}

\usepackage[a4paper,top=2cm,bottom=2cm,left=1.5cm,right=1.5cm,marginparwidth=1.75cm]{geometry}
\usepackage{amsmath}
\usepackage{amsfonts}
\usepackage{amssymb}
\usepackage{amsthm}
\usepackage{graphicx}
\usepackage{booktabs}
\usepackage{url}
\usepackage{authblk}
\usepackage{bm}
\usepackage[authoryear]{natbib}
\usepackage{setspace}

\input{def}

\newtheorem{theorem}{Theorem}

\title{Outcome-blinded sample size re-estimation for externally controlled single-arm trials using baseline covariates}
\author[1]{Keisuke Hanada\thanks{Department of Biostatistics, Wakayama Medical University, Kimiidera, Wakayama, 641-8509, Japan \quad E-mail: khanada@wakayama-med.ac.jp}}
\author[2]{Masahiro Kojima}
\affil[1]{Department of Biostatistics, Faculty of Medicine, Wakayama Medical University}
\affil[2]{Department of Data Science for Business Innovation, Chuo University, Tokyo, Japan}
\date{}

\begin{document}

\maketitle

\begin{abstract}
Externally controlled single-arm trials provide an option when limited patient populations or ethical constraints make concurrent randomized controls impractical, including in rare diseases and investigator-initiated trials with limited recruitment. Standardization improves comparability by aligning external controls with the enrolled population, but does not by itself preserve planned power. Differences between anticipated and enrolled covariate distributions can change the precision of the standardized control estimate and leave a fixed-size trial underpowered.
We propose an outcome-blinded, covariate-adaptive sample size re-estimation procedure that translates these changes in precision into updated recruitment targets. Historical-control data and accumulating active-arm baseline covariates are used to update external-control standardization and its estimated precision. Adaptation requires no active-arm outcomes and retains the prespecified clinically meaningful effect to be detected. We state sufficient conditions for type I error calibration and power after adaptive stopping, assess operating characteristics through simulations, and illustrate implementation using Alzheimer's Disease Cooperative Study data.
In the primary simulations, adaptive recruitment restored power lost under standardized fixed designs, achieving at least the target power across the distribution-shift settings and outcome models examined. Type I error rates were close to nominal, although mild inflation remained with a smaller historical-control sample. The application illustrated how accumulating baseline information guided recruitment revisions. By adapting sample size to the precision of the population-standardized control estimate, the proposed procedure addresses a source of power loss that standardization alone does not resolve.

\end{abstract}

\noindent\textit{Key words and phrases:} historical control; inverse probability weighting; population-standardized estimand; sample size re-estimation; single-arm trial

\setstretch{1.3}

\section{Introduction}
\label{sec:intro}

Randomized controlled trials remain the standard design for evaluating treatment effects. In rare diseases, including those affecting children, small patient populations can make difficult to conduct randomized controlled trials. 
Single-arm trials using external controls offer a potential approach to treatment evaluation in these settings, as discussed in FDA guidance on rare disease drug development \citep{fda2023rare}. 
Such designs may also be considered in severe diseases, exploratory efficacy evaluations, or other settings in which allocation to a concurrent control arm is ethically or practically difficult. 
Externally controlled or historical-control single-arm trials based on existing clinical trials, registries, electronic health records, or other external data sources can provide a practical option for efficacy evaluation in such settings \citep{pocock1976historical, viele2014historical, thorlund2020synthetic, seeger2020external}. 
The growing discussion of real-world data for regulatory decision making and the increasing use of external controls in oncology further underscore the methodological importance of externally controlled trial designs \citep{burcu2020rwe, mishraKalyani2022external}.

The main difficulty in externally controlled comparisons is that the active-arm trial and the external control group are not generated by the same randomization mechanism. Differences in patient characteristics, calendar time, measurement procedures, eligibility criteria, or follow-up may cause a simple difference in outcome means to reflect population or data-source differences as well as a treatment effect. Externally controlled single-arm trials therefore require a clear design-stage definition of the target population and of how external-control information is aligned with that population. Propensity scores are central tools for confounding adjustment in observational studies \citep{rosenbaumRubin1983}, and frameworks for external-control single-arm trials have been developed through G-computation, inverse probability of treatment weighting (IPTW), and doubly debiased machine learning \citep{loiseau2022external}. More broadly, clinically relevant treatment-effect inference requires an explicit target population and appropriate adjustment for differences between that population and the populations contributing evidence \citep{stuart2011generalizability, dahabreh2020transport, hanada2026inmass}.

Sample size determination poses additional challenges in trials using historical controls. Previous work has considered sample size calculations that incorporate the historical-control mean and its uncertainty \citep{makuch1980sample, zhang2010calculating}, sample size determination for covariate-adjusted historical-control comparisons \citep{o2002sample}, and design and analysis considerations when previous trials are used as control information \citep{schoenfeld2019design}. These studies provide an important basis for externally controlled trial design, but they typically focus on settings in which the sample size is fixed at trial initiation or the historical-control mean and variance are treated as planning values. When the external control is standardized to the active-arm population actually enrolled, the standardized control mean and its variance component depend on the active-arm covariate distribution. Thus, if the target population is defined as the enrolled active-arm population, sample size design should also reflect uncertainty about that target population.

Sample size re-estimation (SSR) has developed as an adaptive design approach that updates the required sample size using information accumulated during a trial. 
Methods that update nuisance parameters, such as the variance of a normal outcome, under blinding, and methods that extend a trial based on conditional power have been widely discussed, mainly in randomized controlled trials \citep{gouldShih1992, proschanHunsberger1995, gould2001ssr}. 
Recent work has also studied blinded monitoring of nuisance parameters and blinded SSR that accounts for uncertainty in interim variance estimates \citep{xuFriede2026monitoring, maeda2026blinded}. In hybrid-control designs that combine a current study with a historical study, IPW has been used to evaluate covariate imbalance and to re-estimate sample size under blinding according to that imbalance \citep{kojima2026hybrid}. 

In the single-arm externally controlled design, the baseline covariate distribution observed during enrollment determines the target distribution to which the external control is standardized. 
The quantities to be updated are therefore not only the extent of covariate imbalance but also the standardized control parameter and its variance component. 
If the enrolled population differs from that anticipated at planning, standardization can change both the expected control response and the precision of the historical-control estimate. Consequently, the initially planned sample size may no longer provide the intended power. 
How to re-estimate sample size in response to this evolving standardized control reference remains an unresolved design problem. 
This observation motivates an outcome-blinded SSR procedure that updates the required sample size using only enrolled baseline covariates, without using active-arm outcomes.

This paper makes three main contributions. First, Section~\ref{sec:estimands} defines a population-standardized estimand for externally controlled single-arm trials in which the target population is the actually enrolled active-arm population. 
Second, Section~\ref{sec:ssr} develops an initial sample size design and a covariate-adaptive SSR procedure for this estimand, using only baseline covariate information observed during enrollment. 
Third, we state sufficient conditions for controling type I error and power of the covariate-adaptive SSR procedure. These conditions require separate justification for the final analysis. Section~\ref{sec:numerical} first uses an illustrative real-data application to show the implementation of the standardized control parameter, IPTW weights, and review-specific sample size re-estimation. It then presents simulation studies comparing type I error, power, and final sample size of the proposed method with fixed and oracle designs when the planned and true active-arm distributions differ. Section~\ref{sec:discussion} discusses methodological implications, assumptions and limitations, and possible extensions to time-to-event outcomes and more complex external-control data structures.

\section{Population-Standardized Estimands and an Example}
\label{sec:estimands}

\subsection{Data structure}

We consider a single-arm trial with an external control. Let $D_A=\{(Y_i,\X_i,Z_i): i=1,\ldots,n_A\}$ denote data from the active-arm trial and $D_H=\{(Y_j,\X_j,Z_j): j=1,\ldots,n_H\}$ denote historical-control data. Here, $Y$ is the observed outcome and $\X$ is a vector of baseline covariates measured before treatment. For each subject, let $Y^1$ and $Y^0$ denote the potential outcomes under active treatment and under control or standard treatment, respectively. Thus, $Y=Y^1$ is observed for active-arm subjects, whereas $Y=Y^0$ is observed for controls.

Let $G$ be a binary data-source indicator, where $G=1$ denotes membership in the active-arm trial and $G=0$ denotes membership in the historical-control data source. The variable $G$ is not a treatment indicator but an indicator of data source or trial participation; treatment states in the potential outcome notation are denoted by superscripts $a\in\{0,1\}$. In what follows, the subscript $A$ denotes the active-arm target population, whereas individual-level data-source membership is denoted by $G$. 
We let $Z=s(\X)$ be a pre-specified balancing score that summarizes baseline distributional differences between the active-arm trial and the historical controls. Typically, $s(\X)=P(G=1\mid \X)$, the conditional probability of participation in the active-arm trial, which we refer to as the trial participation score.
This use of a trial participation score parallels propensity score methods for generalizing results from randomized trials to a target population \citep{stuart2011generalizability}. In the present design, the covariates and functional form of $s(\cdot)$ are pre-specified at the design stage, while its coefficients may be refitted at reviews, and sample size determination and re-estimation summarize the enrolled active-arm population through the distribution of $Z$. When IPTW standardization is used, historical-control subjects receive weights based on this balancing score. In particular, if $Z=s(\X)$ is a trial participation score, $Z/(1-Z)$ can be interpreted as an odds-type density-ratio weight that shifts the historical-control distribution toward the active-arm score distribution. Positivity, model specification, and covariate balance after weighting are important for weighting-based standardization \citep{coleHernan2008weights, austinStuart2015iptw}.

\subsection{Target population and estimand}

The target population consists of the subjects actually enrolled in the single-arm trial. Let $F_{\X,A}$ denote its baseline covariate distribution and $F_{Z,A}$ the corresponding trial participation score distribution. The estimand of interest is the difference between the outcome under active treatment in this active-arm population and the potential outcome that the same population would have experienced under control treatment. This formulation follows the estimand framework, which defines a treatment effect by specifying the treatment conditions, population, endpoint, and population-level summary \citep{ichE9R1}. The estimand is not an average effect in the entire external-control population; rather, it is close in spirit to an average treatment effect on the treated for an external-control analysis because it uses the counterfactual control outcome in the actually enrolled active-arm population.

Let $\eta_a(F_{Z,A})$ denote an endpoint-specific outcome functional standardized to the active-arm score distribution, where $a=1$ denotes active treatment and $a=0$ denotes control treatment. For example, $\eta_a(F_{Z,A})$ may be a mean, an event proportion, or a survival probability at time $t$. We define the treatment-effect estimand as
\[
\Delta(F_{Z,A})
=
\eta_1(F_{Z,A})-\eta_0(F_{Z,A}).
\]
For risk ratios, odds ratios, or other measures, $\Delta(F_{Z,A})$ can be defined on the corresponding transformation scale as $\Delta(F_{Z,A})=h\{\eta_1(F_{Z,A}),\eta_0(F_{Z,A})\}$. Under this definition, the target population is the enrolled active-arm population. Historical-control information is therefore standardized to the active-arm score distribution $F_{Z,A}$. This interpretation requires, at a minimum, that conditional on the observed baseline covariates or score, historical-control outcomes under control treatment are representative of the control potential outcomes in the active-arm target population. It also requires sufficient overlap over the score region represented in the active-arm target population. If these conditions are not met, the standardized control parameter may not adequately represent the control potential outcome in the target population.

In what follows, we fix the clinically meaningful design contrast as $\Delta_{\mathrm{design}}(F)\equiv\Delta_{\mathrm{design}}$, where $F$ denotes the target score distribution. Alternatively, a design may fix the active-arm mean outcome at $\eta_1^\ast$ on clinical grounds. In that case, the design contrast is
\[
\Delta_{\mathrm{design}}(F)=\eta_1^\ast-\eta_0(F),
\]
and therefore depends on the target population's covariate distribution. Such a design requires confirming that the design contrast remains clinically meaningful across the range of target distributions considered.
For the fixed-contrast design adopted as the primary approach in this paper, sample size planning prespecifies a clinically meaningful difference $\Delta_{\mathrm{design}}\ne0$ and aims to achieve the desired power under the planning alternative $\eta_1(F)=\eta_0(F)+\Delta_{\mathrm{design}}$, equivalently $\Delta(F)=\Delta_{\mathrm{design}}$, for the test
\[
H_0:\Delta(F)=0
\quad\text{versus}\quad
H_1:\Delta(F)\ne0.
\]

\subsection{Standardized control parameters for sample size design}

For sample size design, the control component $\eta_0(F_{Z,A})$ is particularly important. It is the outcome characteristic expected if the active-arm population had received control treatment, and it serves as the control-side planning value.

Suppose that a conditional control parameter $\theta_0(z)$ can be estimated from the historical-control data on the endpoint-specific scale used for standardization. Depending on the endpoint, $\theta_0(z)$ may represent a mean, event probability, survival probability at a fixed time, restricted mean survival time, variance, or another outcome characteristic. The standardized control parameter for the active-arm score distribution can then be written as
\[
\eta_0(F_{Z,A})
=
\int \theta_0(z)\,dF_{Z,A}(z).
\]

Full standardization based on baseline covariates $\X$ can be defined by expressing the conditional parameter as a function of $\X=\x$. In practical sample size design and interim re-estimation, however, it is often difficult to handle the full distribution of high-dimensional $\X$ stably. We therefore use $\eta_0(F_{Z,A})$ based on the pre-specified score $Z=s(\X)$ as the main design parameter. This design-stage approximation implements target-population standardization through a pre-specified low-dimensional score distribution \citep{stuart2011generalizability, dahabreh2020transport}. In sample size design, both the standardized control parameter $\eta_0(F_{Z,A})$ and the variance of its estimator $\hat\eta_0(F_{Z,A})$ may depend on the target distribution. Accordingly, the sample size formula below uses the variance of the treatment-effect estimator $\hat\Delta(F)=\hat\eta_1(F)-\hat\eta_0(F)$.

\subsection{A simple illustration of target-population-dependent sample size}

We consider a simple example for illustrating the dependence of target population on the sample size determination.
Let $X$ be binary, with historical prevalence $p_H=0.5$ and active-arm target prevalence $p=P(X=1)$. Consider 100 historical controls, with 50 in each covariate stratum, conditional control-outcome variance 1, and active-arm planning variance 1. The standardized control estimator is $(1-p)\bar Y_{H,0}+p\bar Y_{H,1}$. For a fixed target prevalence, its variance is
\[
\tau_0^2(p)=\frac{(1-p)^2}{50}+\frac{p^2}{50}.
\]
We hold the planning contrast fixed at $\Delta_{\mathrm{design}}=0.5$. With two-sided $\alpha=0.05$ and target power 0.80, the required active-arm sample size is 46 at $p=0.5$ and 55 at $p=0.8$. Thus, sample size can increase because standardization reduces the precision of the historical-control estimate, even though the effect to be detected remains unchanged.

\begin{figure}[ht]
\centering
\includegraphics[width=0.9\linewidth]{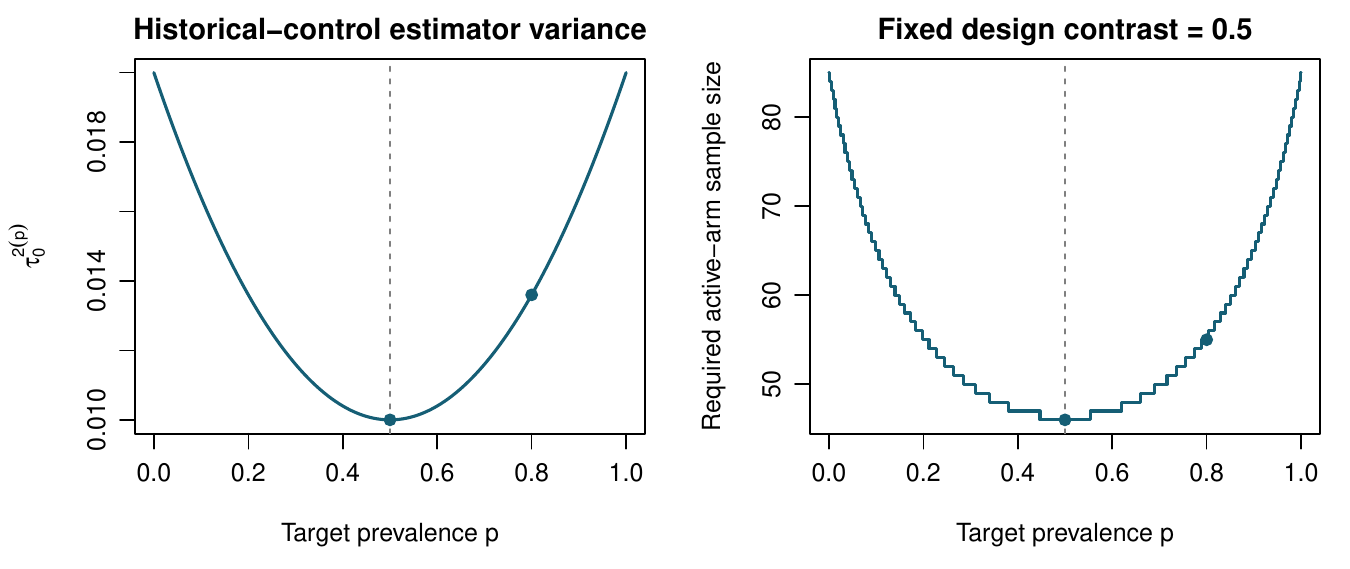}
\caption{Historical-control estimator variance and required active-arm sample size for a fixed design contrast of 0.5. The vertical line indicates the historical prevalence $p_H=0.5$.}
\label{fig:simple-illustration}
\end{figure}

The change in the target distribution can be assessed from accumulating baseline covariates without examining active-arm outcomes. Web Appendix A gives the outcome model and calculation. The example isolates the fixed-target planning calculation; estimating the target distribution and adapting enrollment introduce additional uncertainty requiring assessment in the full design.

\section{Sample Size Determination and SSR}
\label{sec:ssr}

\subsection{Initial design}
\label{sec:initial-design}

At the planning stage, the final active-arm baseline covariate distribution $F_{X,A}$ and score distribution $F_{Z,A}$ are unknown. The initial design therefore pre-specifies an expected active-arm score distribution using previous clinical trials, registries, natural history studies, eligibility criteria, and clinical knowledge. We denote this planning score distribution by $F_{Z,A}^{\mathrm{plan}}$. In a standard setting, one may assume that the active-arm population has a background distribution similar to that of the historical controls and set $F_{Z,A}^{\mathrm{plan}}=F_{Z,H}$.
Under the planning score distribution $F_{Z,A}^{\mathrm{plan}}$, the standardized control parameter is $\eta_0(F_{Z,A}^{\mathrm{plan}})$. This is the outcome expected if the planned active-arm population had received control treatment and serves as the control-side reference value for sample size determination.

We consider a normal approximation test as the basic sample size design. Determining sample size from the design contrast, variance components, significance level, and target power under a normal approximation is widely used in clinical trial sample size planning \citep{chow2017sample}. Let $F$ denote a generic target score distribution and define the treatment-effect estimator as
\[
\hat\Delta(F)=\hat\eta_1(F)-\hat\eta_0(F).
\]
Here, $\eta_1(F)$ is the true active-arm mean in target population $F$, and $\hat\eta_1(F)$ is estimated from active-arm outcomes only at the final analysis. Under the planning alternative, we assume $\eta_1(F)=\eta_0(F)+\Delta_{\mathrm{design}}$.

When the active-arm mean is uncorrelated with the estimated standardized control mean, the basic planning variance is
\[
\operatorname{Var}\{\hat\Delta(F)\}\simeq\frac{\sigma_1^2}{N}+\tau_0^2(F),
\]
where $\sigma_1^2$ is the prespecified active-arm outcome variance and $\tau_0^2(F)$ is the variance of the standardized historical-control mean estimator, including estimation of the target and participation score when these are estimated. This variance also depends on the size $m$ of the covariate sample representing the target.
With IPTW, if the target score distribution $F$ moves away from the historical-control score distribution, poor overlap and increased weight variability can increase $\tau_0^2(F)$.

Based on this variance expression, define the standardized statistic used for sample size design and theoretical arguments as
\[
T_N(F)
=
\frac{
\hat\Delta(F)
}{
\sqrt{\sigma_1^2/N+\tau_0^2(F)}
},
\]
and let its mean under the alternative be
\[
\mu_N(F)
=
\frac{\Delta(F)}
{\sqrt{\sigma_1^2/N+\tau_0^2(F)}}.
\]
The actual final analysis uses the corresponding plug-in test statistic with estimated standardized control parameters and standard errors. The sample size design below assumes that $T_N(F)-\mu_N(F)$ is approximately standard normal. For two-sided level $\alpha$ and target power $1-\beta$, planning evaluates this approximation at $\Delta(F)=\Delta_{\mathrm{design}}$ and chooses $N$ so that
\[
\frac{|\Delta_{\mathrm{design}}|}
{
\sqrt{\sigma_1^2/N+\tau_0^2(F)}
}
\geq
z_{1-\alpha/2}+z_{1-\beta}.
\]
Therefore, the required sample size is
\begin{equation}
\label{eq-NF}
N(F)
=
\left\lceil
\frac{
\sigma_1^2
}{
\Delta_{\mathrm{design}}^2/
\{z_{1-\alpha/2}+z_{1-\beta}\}^2
-
\tau_0^2(F)
}
\right\rceil .
\end{equation}
If the denominator is not positive, the planning target power cannot be achieved by increasing active-arm enrollment alone under this approximation. Operationally, the prespecified enrollment cap is used, without a claim of achieving target power. For a one-sided test, $z_{1-\alpha/2}$ is replaced by $z_{1-\alpha}$. The planning-stage sample size is then $N_{\mathrm{plan}}=N(F_{Z,A}^{\mathrm{plan}})$.

\subsection{Outcome-blinded covariate-adaptive SSR procedure}
\label{sec:ssr-procedure}

During a single-arm trial, baseline covariates $X$ in the active arm are usually observed earlier than outcomes. The proposed method is an outcome-blinded SSR procedure that updates the standardization target using only the score distribution of enrolled subjects and recalculates the required sample size using Equation~\eqref{eq-NF}. Thus, differences between the planning score distribution and the actual enrolled population can be reflected in the sample size design without using unblinded treatment-effect information.

The implementation can be summarized as follows.

\begin{enumerate}
    \item \textbf{Pre-specification.}
    Pre-specify $\Delta_{\mathrm{design}}$, $\sigma_1^2$, the covariates and score-model structure $s(\cdot)$, the planning sample size $N_{\mathrm{plan}}$, the lower bound $N_{\min}$ and maximum sample size $N_{\max}$ allowed for the target sample size, the minimum review size $n_0$, the set of review times $\mathcal{M}$, and whether re-estimation uses the point estimate or a conservative upper bound.

    \item \textbf{Update of the score distribution.}
    At each review time $n\in\mathcal{M}$, refit the prespecified score model using historical and accumulated active-arm baseline covariates, calculate the scores, and obtain their empirical active-arm distribution $\hat F_{Z,A,n}$.

    \item \textbf{Standardization and recalculation of required sample size.}
    Standardize the historical-control information to $F=\hat F_{Z,A,n}$ and calculate $N_{\mathrm{re}}(n)=N(\hat F_{Z,A,n})$ using $N(F)$ from the previous subsection with the fixed $\Delta_{\mathrm{design}}$ and estimated historical-control variance. In the point-estimate SSR, set $N_{\mathrm{re}}^{\dagger}(n)=N_{\mathrm{re}}(n)$. In the conservative SSR, set $N_{\mathrm{re}}^{\dagger}(n)$ to the upper limit of the uncertainty interval described below.

    \item \textbf{Target sample size and stopping decision.}
    Update
    \[
    N_{\mathrm{target}}(n)
    =
    \min\left\{
    N_{\max},
    \max\left(N_{\min},N_{\mathrm{re}}^{\dagger}(n)\right)
    \right\}.
    \]
    If sample size reduction is not allowed, use $N_{\mathrm{plan}}$ instead of $N_{\min}$ as the lower bound. If $n\geq N_{\mathrm{target}}(n)$, stop enrollment and set $N_{\mathrm{final}}=n$. Otherwise, retain the target between reviews and stop when it is reached, even if this occurs before the next scheduled review. At a scheduled review, update the target before deciding whether to stop. Enrollment also stops at $n=N_{\max}$.
\end{enumerate}

Here, $\mathcal{M}\subseteq\{n_0,n_0+1,\ldots,N_{\max}\}$ is the review-time set, and $n_0$ is the minimum review size used to avoid highly unstable estimation of the score distribution. In the extreme case, one can set $\mathcal{M}=\{n_0,n_0+1,\ldots,N_{\max}\}$ and update the required sample size after every enrolled subject.

The conservative SSR reflects the finite-sample uncertainty of $\hat F_{Z,A,n}$ while respecting the positivity and right skewness of sample size. We estimate the standard error of $\log N(\hat F_{Z,A,n})$ using the delta method. Let $\widehat{\mathrm{se}}\{\log N_{\mathrm{re}}(n)\}$ denote this estimated standard error. An uncertainty interval for the required sample size is constructed as
\[
\left[
\exp\left\{\log N_{\mathrm{re}}(n)-z_{1-\gamma/2}\widehat{\mathrm{se}}(\log N_{\mathrm{re}}(n))\right\},
\exp\left\{\log N_{\mathrm{re}}(n)+z_{1-\gamma/2}\widehat{\mathrm{se}}(\log N_{\mathrm{re}}(n))\right\}
\right],
\]
where $\gamma$ is the significance level for the uncertainty interval. Working on the log scale ensures that both endpoints are positive. Let the upper endpoint be $N_{\mathrm{re}}^{\mathrm{U}}(n)$; the conservative SSR uses $N_{\mathrm{re}}^{\dagger}(n)=\lceil N_{\mathrm{re}}^{\mathrm{U}}(n)\rceil$. 

Outcome blinding does not itself ensure type I error control. Prognostic covariates determine adaptation and can be associated with final outcomes; historical outcomes also contribute to both the sample-size rule and the final statistic. The final analysis therefore requires calibration for the joint estimation and stopping procedure. 
We give the sufficient conditions of controlling type I error and power in the SSR procedure in Web Appendix B. 
These conditions require that the type I error rate be controlled at the design stage, and that the difference between the test statistic and its mean at the final analysis be sufficiently close to a standard normal distribution.
In other words, the underlying test must be appropriate for evaluating the targeted hypothesis.

\section{Numerical Studies}
\label{sec:numerical}

\subsection{Illustrative application}

We illustrate the implementation of the proposed population-standardized design and outcome-blinded SSR using ADCS data. The goal is to show how historical controls are standardized to the enrolled active-arm population and how the required sample size is updated at scheduled reviews. 

The control arm of ADC-016 was regarded as a historical-control cohort, whereas the treatment arm of ADC-027 was regarded as a hypothetical active-arm single-arm trial. ADC-016 was a randomized trial of high-dose B vitamin supplementation in Alzheimer disease, and ADC-027 was a randomized trial of docosahexaenoic acid supplementation in Alzheimer disease \citep{aisen2008bvitamin, quinn2010dha}. ADCS trial data have also been used to illustrate methods for incorporating historical controls in clinical trials, including Bayesian borrowing approaches for longitudinal outcomes \citep{qi2022historical}. During sample size re-estimation, outcomes in the ADC-027 treatment arm were treated as masked, and only accumulated baseline covariates were used to update the target distribution, standardized historical-control mean, and required sample size.

Let $F$ denote a generic target score distribution. From the trial participation score model corresponding to $F$, let $\hat Z_j^H(F)$ be the estimated score for historical-control subject $j$, and define the IPTW weight
\[
w_j(F)=\frac{\hat Z_j^H(F)}{1-\hat Z_j^H(F)}.
\]
The standardized control mean is then estimated by
\[
\hat\eta_0(F)
=
\frac{\sum_{j\in H}w_j(F)Y_j^0}
{\sum_{j\in H}w_j(F)}.
\]
Let $F$ denote a generic target score distribution. From the trial participation score model corresponding to $F$, let $\hat Z_j^H(F)$ be the estimated score for historical-control subject $j$, and define the IPTW weight
\[
w_j(F)=\frac{\hat Z_j^H(F)}{1-\hat Z_j^H(F)}.
\]
The standardized control mean is then estimated by
\[
\hat\eta_0(F)
=
\frac{\sum_{j\in H}w_j(F)Y_j^0}
{\sum_{j\in H}w_j(F)}.
\]
The unadjusted fixed design performs no standardization and uses the historical-control sample mean as the control parameter. The population-standardized fixed design standardizes to the score distribution based on active-arm covariate data assumed to be available at planning and fixes the sample size at trial initiation. In the proposed SSR designs, the trial participation score model is re-estimated at each review time using all historical controls and the baseline covariates of active-arm subjects enrolled up to that review. We then set $F=\hat F_{Z,A,n}$ and update $\hat\eta_0(F)$ and the required sample size. The point-estimate version uses $N_{\mathrm{re}}(n)$, whereas the conservative version uses an upper limit that accounts for estimation uncertainty. The numerical designs are summarized in Table~\ref{tab:simulation-designs}, which include an oracle and a max-scenario fixed designs using  in the simulation study.

\begin{table}[ht]
\small
\centering
\caption{Designs compared in the numerical studies.}
\label{tab:simulation-designs}
\begin{tabular}{p{0.28\linewidth}p{0.24\linewidth}p{0.28\linewidth}}
\toprule
Design & Target used for planning & Sample size rule \\
\midrule
Unadjusted fixed &
No weighting &
Fixed at planning \\
Population-standardized fixed &
$F_{Z,A}^{\mathrm{plan}}$ &
Fixed at planning \\
Oracle &
$F_{Z,A}^{\mathrm{true}}$ &
Fixed oracle benchmark \\
Max-scenario fixed &
Each $F\in\mathcal{F}_{Z,A}$ &
Fixed at $\max_{F\in\mathcal{F}_{Z,A}}N(F)$ \\
Proposed SSR, point estimate &
$\hat F_{Z,A,n}$ &
Update with $N_{\mathrm{re}}^{\dagger}(n)=N_{\mathrm{re}}(n)$ \\
Proposed SSR, conservative &
$\hat F_{Z,A,n}$ &
Update with $N_{\mathrm{re}}^{\dagger}(n)=N_{\mathrm{re}}^{\mathrm{U}}(n)$ \\
\bottomrule
\end{tabular}
\end{table}

Case 1 used sex and education to estimate the participation score; Case 2 additionally used age and baseline MMSE. We fixed the design contrast magnitude at $|\Delta_{\mathrm{design}}|=6$ ADAS-Cog points, the active-arm variance at $13^2$, two-sided $\alpha=0.05$, and target power at 0.80. The historical cohort contained 138 subjects and the active covariate pool 232. Web Appendix C describes the cohorts and review schedules.

Table~\ref{tab:real_data_design} reports the variance estimates, recruitment targets, and final enrollment. As in the simulations, the standardized designs use the joint sandwich variance in Equation~\eqref{eq:joint-variance}, incorporating estimation of the participation score from historical and active-arm covariates. The unadjusted fixed design required 49 subjects in both cases. The population-standardized fixed design required 50 in Case 1 and 49 in Case 2; it uses all 232 active covariate records retrospectively and is a reference comparator rather than a prospectively available design.

Point and conservative SSR both enrolled 51 subjects in Case 1, and enrolled 64 and 71, respectively, in Case 2. Case 1 was reviewed at 20, 30, and 40 subjects; Case 2 was reviewed at 20, 40, and 60 and stopped before the scheduled review at 80. In Case 2, the estimated historical variance decreased from 3.8648 at the first review to 1.9430 at the last; the initial target reached the cap of 232 but was subsequently revised downward. With the effect to be detected held fixed, updates in the point requirement reflect changes in the historical-control variance estimate. The conservative rule additionally reflects local sensitivity of the calculated requirement. In Case 1, the finite-difference perturbations did not change the rounded required sizes, so this rule added no enrollment.

\input{real_data_design_table}

\subsection{Simulation study}

The simulation study evaluated fixed-contrast sample size re-estimation when the active-arm covariate distribution differed from the planning distribution. We examined rejection probabilities under a generating-population null hypothesis and a fixed marginal alternative, together with final sample size. The implemented final test and historical-control variance estimator were retained without an additional finite-sample correction.

\paragraph{Model settings}

Historical-control and active-arm covariates were generated from $X\mid G=0\sim N(0,1)$ and $X\mid G=1\sim N(\mu_A,1)$, respectively, with $\mu_A\in\{0,0.5,1.0\}$. Planning assumed $F_{X,A}^{\mathrm{plan}}=N(0,1)$. For outcome generation, the score corresponding to equal prior group probabilities was
\[
z_{\mu_A}(x)=\operatorname{expit}(\mu_A x-\mu_A^2/2).
\]
Historical outcomes were generated independently conditional on covariates as $Y^0\mid X=x\sim N\{\theta_0(z_{\mu_A}(x)),1\}$. The main linear model used $\theta_0(z)=0.5+z$. Additional models were $\theta_0(z)=0.5+z+0.5z^2$ and $\theta_0(z)=0.5+I(z>0.5)$; their results are reported in Web Appendix D. Thus, the outcome-generating score and the induced conditional outcome model were defined separately within each $\mu_A$ setting.

The generating-population control mean was evaluated by numerical integration,
\[
\eta_0(F_{Z,A}^{\mathrm{true}})
=\int\theta_0\{z_{\mu_A}(x)\}\phi(x-\mu_A)\,dx,
\]
where $\phi$ is the standard normal density. Active-arm outcomes were independent of their own covariates and were generated from $N\{\eta_0(F_{Z,A}^{\mathrm{true}}),1\}$ under the null and $N\{\eta_0(F_{Z,A}^{\mathrm{true}})+0.5,1\}$ under the alternative. Each replicate shared its historical data, active covariate pool, and design decisions between null and alternative evaluations. Independent null and alternative outcome pools were generated, with each design using the prefix corresponding to its final sample size.

All six designs used the fixed planning contrast $\Delta_{\mathrm{design}}=0.5$ and planning variance $\sigma_1^2=1$. The primary simulations used $n_H=500$ historical controls to assess basic performance with a larger historical information base; a sensitivity analysis used $n_H=150$ to examine performance with more limited historical information. In both sets, the two-sided significance level was $0.05$, and target power was $0.80$. The true alternative contrast was also $0.5$. The trial participation score used for estimation was fitted by logistic regression of group membership on $X$; it was not replaced by the generating score. Standardized historical means were unchanged, and their joint sandwich variances propagated participation-score and target-sample estimation as detailed in Web Appendix D.

\paragraph{Comparator designs and review schedule}

We compared the six designs in Table~\ref{tab:simulation-designs}. At planning, independent reference samples of size 1,000 represented the planning distribution, the true active distribution for the oracle design, and each candidate distribution $N(\mu,1)$ for $\mu\in\{0,0.5,1.0\}$. Required sample sizes were rounded upward and restricted to $[20,n_H]$. The minimum candidate requirement supplied the lower SSR bound, while the absolute SSR cap was $N_{\max}=n_H$, equal to 500 in the primary simulations and 150 in the sensitivity analysis. The lower bound was recalculated using the fixed planning contrast in each replicate; the cap did not depend on the candidate requirements. The max-scenario fixed design used the maximum candidate sample size. The oracle used the true active covariate distribution at planning, represented by a finite reference sample of 1,000. It therefore retained reference-sampling and score-fitting variability as well as estimation of the historical-control mean and variance.

In SSR, the first review was at 20 subjects and subsequent scheduled reviews were at increments of 10 up to the enrollment bound. At each review, the participation score was refitted using historical data and accumulated active-arm covariates, and the historical-control mean and variance were updated. The point procedure recalculated the required sample size with the fixed contrast. The conservative procedure used the 90\% upper limit for log sample size obtained by the delta method, with a deterministic reference sample of size 500 for the numerical derivative. Both targets were restricted by the candidate-based lower bound and the independently specified absolute cap $n_H$. Enrollment stopped immediately if the target was no larger than the current enrollment; if the target fell before the next scheduled review, enrollment continued only to that target. When enrollment ended between scheduled reviews, the final analysis retained the historical-control estimates from the last review. The standardized fixed designs instead fitted the participation score using the final active covariate sample for their final analyses.

Each setting used 10,000 replicates. We report the mean and standard deviation of final enrollment, null rejection probability, and power. The Monte Carlo standard error is approximately 0.0022 for a rejection probability of 0.05 and 0.0040 for a power of 0.80.

\paragraph{Results}

Tables~\ref{tab:sim-linear}--\ref{tab:sim-threshold} report the primary results with 500 historical controls. In the linear model at $\mu_A=1$, the population-standardized fixed design had power 0.747 and mean enrollment 34.0. Point and conservative SSR increased mean enrollment to 41.3 and 45.5, respectively, with power 0.810 and 0.832. The oracle had power 0.787 and mean enrollment 39.5. Corresponding type I error rates were 0.0514 for point SSR, 0.0511 for conservative SSR, and 0.0525 for the oracle. Thus, standardization alone did not prevent power loss when the planning and enrolled covariate distributions differed; updating enrollment recovered power in this setting.

At $\mu_A=1$, the quadratic and threshold models showed standardized fixed-design powers of 0.721 and 0.728. Point SSR increased mean enrollment from 34.0 and 34.2 to 44.2 and 45.3, with powers of 0.809 and 0.816; conservative SSR had mean enrollment 49.0 and 49.1 and powers of 0.835 and 0.836. Smaller distributional shifts also warranted increases. At $\mu_A=0.5$, mean enrollment across the three models was 34.0--34.2 for the standardized fixed design, 35.5--40.7 for point SSR, and 36.7--41.5 for conservative SSR. Corresponding powers were 0.753--0.793, 0.806--0.810, and 0.819--0.821. The amount of additional enrollment therefore depended on both the distributional shift and the outcome model.

Across the standardized designs and settings with 500 historical controls, observed null rejection probabilities ranged from 0.0456 to 0.0540. For threshold-model point SSR at $\mu_A=1$, analytically integrating over the independent normal active-arm mean gave a null rejection estimate of 0.0501 (Monte Carlo SE 0.0005). At $\mu_A=1$, the actual cap of 500 was reached in only 0.07--0.09\% of point-SSR runs and 0.11--0.17\% of conservative-SSR runs across the three models.

The unadjusted fixed design had inflated null rejection probabilities when the covariate distributions differed, because its historical mean did not estimate the population target mean. Its high alternative rejection probabilities should therefore not be interpreted as evidence of a well-calibrated design. The sensitivity analysis with 150 historical controls is reported in Web Appendix D (Tables~\ref{tab:sim-linear-n150}--\ref{tab:sim-threshold-n150}). Mild inflation remained in that analysis: at $\mu_A=1$, the two SSR rules had null rejection probabilities of 0.0546--0.0574 across the three models.

\input{simulation_table_linear}
\input{simulation_table_quadratic}
\input{simulation_table_threshold}

\section{Discussion}
\label{sec:discussion}

We developed a framework for sample size determination and covariate-adaptive re-estimation in externally controlled single-arm trials. Historical outcomes are standardized to the active-arm population, and the primary design holds the clinically meaningful planning contrast fixed. Consequently, changes in the required sample size reflect changes in the precision of the standardized historical-control estimate.

The procedure uses historical outcomes and accumulating active-arm baseline covariates without examining active-arm outcomes. The primary simulations with 500 historical controls illustrate complementary roles for standardization and SSR: standardization addresses population differences, whereas enrollment updates recover power lost through changes in precision. Both moderate and larger distributional shifts required additional enrollment. At $\mu_A=1$, point SSR had power 0.809--0.816 and conservative SSR 0.832--0.836, with the absolute cap rarely reached.

Whereas the type I error rates with 500 historical controls were generally close to the nominal level across the settings evaluated, those with 150 historical controls showed mild residual inflation even after joint variance estimation. This residual inflation may reflect estimation error arising from the small historical control sample. Thus, ensuring a sufficiently large historical control sample is important. However, in settings such as rare diseases, the limited number of patients may mean that only a small historical control sample is available. In such cases, it may be necessary either to accept the resulting uncertainty or to develop a test statistic that more fully accounts for uncertainty in estimates derived from the historical controls.

The limitations concern external-control comparability and feasibility. Conditional exchangeability, consistent endpoint definitions, and adequate overlap are essential; unmeasured confounding or calendar-time differences can invalidate the standardized control reference. Extreme weights and score-model misspecification can affect both the estimated mean and its precision. Historical-control uncertainty places a limit on the power attainable by increasing active-arm enrollment. This constraint is particularly relevant in rare diseases, where recruitment may be limited. An enrollment cap can prevent target power from being reached and should be accompanied by a prespecified feasibility decision. A fixed planning contrast also does not assert that the true effect is constant across populations; power under heterogeneous effects requires separate evaluation.

Further work should evaluate the joint variance and final-test calibration under prognostic active-arm covariates, score-model misspecification, and adaptive stopping. When active outcomes depend on the covariates used for standardization, the covariance between the active and standardized historical mean estimators must also be addressed. Hybrid-control designs would additionally require coordinating standardization, external-information borrowing, and recruitment \citep{kojima2026hybrid}. Extensions to non-normal outcomes and alternative estimand strategies require appropriate analysis and power calculations rather than direct reuse of the present normal approximation \citep{ichE9R1}.

\section*{Acknowledgments}
This work was supported by JSPS KAKENHI (Grant numbers JP26K21184 and JP26K21185).

\section*{Data and Code Availability}
The illustrative application uses data from the Alzheimer's Disease Cooperative Study at UC San Diego (\url{https://www.adcs.org/}). R code and execution instructions for the proposed procedures, simulation study, and illustrative application are provided in the supplementary material. Reproducing the illustrative application requires access to the original Alzheimer's Disease Cooperative Study data; a synthetic-data example is included to demonstrate implementation.

\bibliographystyle{biom}
\bibliography{main}

\clearpage

\begin{center}
\LARGE
Supplementary Material for\\
Covariate-Adaptive Sample Size Re-estimation for Population-Standardized Historical Control Designs in Single-Arm Trials
\end{center}

\begin{center}
\large
Keisuke Hanada and Masahiro Kojima
\end{center}

\appendix
\setcounter{figure}{0}
\setcounter{table}{0}
\setcounter{equation}{0}
\setcounter{theorem}{0}
\renewcommand{\thefigure}{S\arabic{figure}}
\renewcommand{\thetable}{S\arabic{table}}
\renewcommand{\theequation}{S\arabic{equation}}
\renewcommand{\thetheorem}{S\arabic{theorem}}

This supplementary material contains four components. Appendix A provides details of the simple illustration of target-population-dependent sample size determination. Appendix B presents the assumptions, theoretical results, and proofs concerning type I error control and conditional power. Appendix C describes the ADCS data, trial participation score models, and review times used in the real-data example. Appendix D presents simulation implementation details and the sensitivity analysis with 150 historical controls. The primary results with 500 historical controls are presented in the main text.

\section{Details of the Simple Illustration}

Let $X$ be binary and consider the control outcome model $Y^0=0.25(X-0.5)+\epsilon$, with independent $\epsilon\sim N(0,1)$. The historical sample contains 50 subjects in each stratum, giving $n_H=100$ and $p_H=0.5$. For target prevalence $p$, the control mean is $\eta_0(p)=0.25(p-0.5)$ and its estimator is
\[
\hat\eta_0(p)=(1-p)\bar Y_{H,0}+p\bar Y_{H,1}.
\]
Conditioning on the historical stratum counts and treating $p$ as fixed gives
\[
\tau_0^2(p)=\frac{(1-p)^2}{50}+\frac{p^2}{50}.
\]
Under the planning alternative, $\eta_1(p)=\eta_0(p)+0.5$ and the active-arm outcome variance is 1. Thus the design contrast is always 0.5, and
\[
N(p)=\left\lceil\frac{1}{0.5^2/(z_{0.975}+z_{0.80})^2-\tau_0^2(p)}\right\rceil.
\]
At $p=0.5$, $\tau_0^2(p)=0.0100$ and $N(p)=46$; at $p=0.8$, $\tau_0^2(p)=0.0136$ and $N(p)=55$. The variance is minimized when the target and historical stratum proportions agree. This calculation treats the target prevalence and conditional outcome variances as known planning quantities.

\section{Theoretical Properties}
\label{sec:web-theory}
\label{sec:supp-theory}

Let $H$ denote the historical-control data and $\mathcal P$ any reference samples generated at planning. Let $A_n=(\X_1,\ldots,\X_n)$ denote accumulated active-arm baseline covariates and $\mathcal G_n=\sigma(H,\mathcal P,A_n)$ the information available without active-arm outcomes. At a scheduled review $r$, the participation-score model is refitted using $H$ and $A_r$, giving the estimated target distribution $\hat F_r$. For each enrollment size $n$ at or beyond the first review, let $r(n)$ be the latest scheduled review at or before $n$ and define $\widetilde F_n=\hat F_{r(n)}$. The target distribution, fitted control mean and variance, and required size are retained between reviews. Write $N_{\mathrm{re}}(n)$ for the requirement from that latest review and set
\[
N_{\mathrm{target}}(n)=\min\{N_{\max},\max(N_{\min},\lceil N_{\mathrm{re}}^\dagger(n)\rceil)\}.
\]
Here $N_{\mathrm{re}}^\dagger(n)$ denotes the point requirement for point updating and its approximate pointwise upper uncertainty limit for conservative updating. The enrollment bounds are fixed before recruitment using a prespecified rule. 
Let $F_{\mathrm{true}}$ be the fixed score distribution specified by the active-arm data-generating mechanism, and write $\Delta(F)=\eta_1(F)-\eta_0(F)$ for the true treatment effect.

For a target distribution $F$ and its associated fitted control estimates, define the actual analysis statistic $\widehat T_N(F)=\hat\Delta(F)/\widehat{\operatorname{se}}\{\hat\Delta(F)\}$. For the adaptive analysis, $\widehat T_n(\widetilde F_n)$ combines these retained control estimates with outcomes from all $n$ active-arm participants. In the formulas below, $\tau_0^2(\widetilde F_n)$ refers to the variance associated with the target sample size $r(n)$, with that size suppressed in the notation. The planning approximation uses
\[
\mu_N(F)=\frac{\Delta(F)}{\sqrt{\sigma_1^2/N+\tau_0^2(F)}}.
\]

\subsection{Assumptions}
\label{sec:supp-assumptions}
\begin{enumerate}
\item[(A1)] The adaptation and stopping rule are measurable with respect to $\mathcal G_n$ and use no active-arm outcomes. The final efficacy analysis is conducted once, after enrollment.
\item[(A2)] Under the fixed data-generating null $H_0^{\mathrm{true}}:\Delta(F_{\mathrm{true}})=0$, let $R_n=\{|\widehat T_n(\widetilde F_n)|>z_{1-\alpha/2}\}$ be the rejection event for the analysis at size $n$. For every admissible stopping size, assume
\[
P_{H_0^{\mathrm{true}}}(R_n\mid\mathcal G_n)\leq\alpha+\rho_n,
\]
where $\rho_n\geq0$ is $\mathcal G_n$-measurable.
\item[(A3)] Under the alternative, write $N=N_{\mathrm{final}}$. Let $\mathcal S_N=(\widetilde F_N,r(N),N,N_{\max})$ denote the stopped design. Conditional on $\mathcal S_N$, assume that the actual final statistic satisfies
\[
\sup_t\left|P\{\widehat T_N(\widetilde F_N)-\mu_N(\widetilde F_N)\leq t\mid\mathcal S_N\}-\Phi(t)\right|\leq\varepsilon(\mathcal S_N).
\]
\item[(A4)] At every non-cap stopped design to which the power claim applies, the target size covers the oracle requirement for the retained target almost surely under the corresponding conditional law, and its denominator is positive:
\[
N_{\mathrm{target}}(N)\geq
\frac{\sigma_1^2}{\Delta(\widetilde F_N)^2/c_\alpha^2-\tau_0^2(\widetilde F_N)},\qquad c_\alpha=z_{1-\alpha/2}+z_{1-\beta}.
\]
\end{enumerate}

\subsection{Type-I error}
\label{sec:supp-type1}
\begin{theorem}[Conditional-calibration bound]
\label{thm:type1}
Under (A1)--(A2), the final procedure satisfies
\[
P_{H_0^{\mathrm{true}}}(R_{N_{\mathrm{final}}})
\leq \alpha+E_{H_0^{\mathrm{true}}}\{\rho_{N_{\mathrm{final}}}\}.
\]
\end{theorem}
\begin{proof}
Since the stopping rule is $\mathcal G_n$-measurable, $\{N_{\mathrm{final}}=n\}\in\mathcal G_n$. Therefore, by the tower property and the conditional calibration assumption,
\[
\begin{aligned}
P(R_{N_{\mathrm{final}}})
&=\sum_nE\left[I(N_{\mathrm{final}}=n)P(R_n\mid\mathcal G_n)\right]\\
&\leq\sum_nE\left[I(N_{\mathrm{final}}=n)(\alpha+\rho_n)\right]
\\
&=\alpha+E\{\rho_{N_{\mathrm{final}}}\}.
\end{aligned}
\]
\end{proof}

\subsection{Conditional power}
\label{sec:supp-power}
\begin{theorem}[Power for the stopped target]
\label{thm:power}
Under (A1), (A3), and (A4), at a non-cap stop,
\[
P\{R_{N_{\mathrm{final}}}\mid\mathcal S_{N_{\mathrm{final}}}\}
\geq 1-\beta-\varepsilon(\mathcal S_{N_{\mathrm{final}}}).
\]
\end{theorem}
\begin{proof}
For a favorable positive effect, the two-sided rejection event contains $\{\widehat T_N>z_{1-\alpha/2}\}$. Conditional on the stopped design, the normal-approximation assumption gives
\[
P(\widehat T_N>z_{1-\alpha/2}\mid\mathcal S_N)\geq 1-\Phi\{z_{1-\alpha/2}-\mu_N\}-\varepsilon(\mathcal S_N).
\]
At a non-cap stop, the stopping rule gives $N\geq N_{\mathrm{target}}(N)$. Together with oracle coverage, this implies $\mu_N\geq z_{1-\alpha/2}+z_{1-\beta}$, so the right side is at least $1-\beta-\varepsilon(\mathcal S_N)$. The negative-effect case follows by reversing signs. The statement applies when the positive-denominator oracle requirement is covered and the enrollment cap does not bind.
\end{proof}

The result is a sufficient-condition statement for the actual final statistic using the retained review estimates. Its conditional approximation includes selection through the stopping rule and the uncertainty in the fitted analysis. It does not claim that point-estimate re-estimation satisfies oracle coverage in every finite sample. An upper uncertainty limit can make coverage more plausible, at the cost of larger enrollment.

Under the planning alternative $\Delta(\widetilde F_N)=\Delta_{\mathrm{design}}$, the oracle requirement in (A4) reduces to the fixed-contrast planning formula with the true variance components. Outside that alternative, the power bound still requires the true effect in (A4).
Blinded SSR provides motivation \citep{gouldShih1992,friedeKieser2006internal,gould2001ssr}, but validity of this external-control procedure requires separate verification of its final analysis and stopping rule.

\section{Details of the Real-Data Example}
\label{sec:real-data-appendix}

ADC-016 placebo recipients were used as the historical-control cohort, and ADC-027 DHA recipients were treated as a hypothetical active-arm single-arm trial \citep{aisen2008bvitamin,quinn2010dha}. The endpoint was the 18-month ADAS-Cog score. The historical cohort included 138 participants with observed endpoints and complete adjustment covariates. The DHA arm had 238 participants; after excluding six records with invalid age, 232 participants had complete baseline covariates for the planning analysis. Active-arm outcomes were not used in the planning calculations. Among the 232 participants, 172 had observed 18-month outcomes; these outcomes are summarized descriptively in Table~\ref{tab:adcs-summary} only.

The trial participation score was estimated by logistic regression, with $G=1$ indicating the DHA cohort and $G=0$ the historical cohort. Case 1 included sex and education,
\[
\operatorname{logit}\{P(G=1\mid\X)\}
=\gamma_0+\gamma_1\,\mathrm{Male}+\gamma_2\,\mathrm{Education},
\]
and Case 2 included age, sex, education, and baseline MMSE,
\[
\operatorname{logit}\{P(G=1\mid\X)\}
=\gamma_0+\gamma_1\,\mathrm{Age}+\gamma_2\,\mathrm{Male}
+\gamma_3\,\mathrm{Education}+\gamma_4\,\mathrm{MMSE}.
\]
At each review, the score model was refitted using all historical covariates and the accumulated active-arm baseline covariates. Case 1 reviews were scheduled at 20, 30, and 40 participants; Case 2 reviews were scheduled at 20, 40, 60, and 80 participants. Enrollment stopped when the current target was reached, so a scheduled review later than the stopping size was not performed. The population-standardized fixed design used all 232 active-arm covariate records retrospectively and serves as a reference.

For this illustration, the planning effect was a fixed absolute difference of 6 ADAS-Cog points, the active-arm variance was set to $13^2$, the two-sided nominal significance level was 0.05, and target power was 0.80. The minimum enrollment was 20 and the cap was 232. At each review, the required size was calculated using the fixed contrast and the current estimate of the standardized historical-control variance. The conservative rule used a one-sided 90\% finite-difference uncertainty adjustment to the calculated requirement. 

The standardized real-data targets use the joint sandwich variance in Equation~\eqref{eq:joint-variance}, with the multivariate design vector corresponding to each case's logistic model. Here $n_H=138$; $m=232$ for retrospective fixed-design planning and $m$ equals the number of accumulated active-arm covariate records at each SSR review. The current-$m$ variance is inserted into Equation~\eqref{eq-NF} without extrapolation to a future target-sample size. The same estimator is refitted for every finite-difference perturbation in the conservative rule.

\begin{table}[ht]
\centering
\caption{Cohort summary for the ADCS illustration.}
\label{tab:adcs-summary}
\begin{tabular}{lcc}
\toprule
 & Historical control & Active-arm covariate pool\\
 & ADC-016 placebo & ADC-027 DHA\\
\midrule
$N$ with complete baseline covariates & 138 & 232\\
Age (years) & 77.2 (7.9) & 75.9 (9.0)\\
Male (\%) & 44.9 & 53.4\\
Education (years) & 13.7 (3.1) & 14.3 (2.9)\\
Baseline MMSE & 21.1 (3.6) & 20.9 (3.6)\\
18-month ADAS-Cog & 29.0 (12.5) [10.0] & 32.8 (13.3) [10.2]\\
\bottomrule
\end{tabular}

{\footnotesize Values are mean (SD) unless otherwise indicated. The historical-control endpoint summary uses all 138 participants; the active-arm endpoint summary uses the 172 participants with observed 18-month outcomes. Bracketed values are residual SDs from a regression adjusted for age, sex, education, and baseline MMSE and are descriptive only.}
\end{table}

\section{Additional Simulation Settings and Results}
\label{sec:additional-simulation}

Let $H$ contain $n_H$ historical records and $A_m$ contain the $m$ target covariate records used to fit the participation score. Let $q_i$ denote the logistic design vector, $(1,X_i)^\top$ in the simulations and $(1,\X_i^\top)^\top$ for the multivariate real-data models, and write $p_i=\operatorname{expit}(q_i^\top\hat\gamma)$ for the fitted, untruncated source-membership probability. For historical observations, let $\tilde p_j$ be $p_j$ truncated to $[\epsilon,1-\epsilon]$, $w_j=\tilde p_j/(1-\tilde p_j)$, and $W=\sum_{j\in H}w_j$. The standardized mean remains
\[
\hat\eta_0(F)=W^{-1}\sum_{j\in H}w_jY_j^0.
\]
To propagate joint estimation of this mean and the logistic score, define
\[
\mathcal I=\sum_{i\in H\cup A_m}p_i(1-p_i)q_iq_i^\top,
\qquad
g=\sum_{j\in H}d_j w_j\{Y_j^0-\hat\eta_0(F)\}q_j,
\]
where $d_j=I(\epsilon<p_j<1-\epsilon)$ is the derivative indicator for truncation. The empirical influence contributions, before division by $W$, are
\begin{align*}
u_{Hj}&=w_j\{Y_j^0-\hat\eta_0(F)\}-p_jq_j^\top\mathcal I^{-1}g,\\
u_{Ai}&=(1-p_i)q_i^\top\mathcal I^{-1}g.
\end{align*}
These follow by linearizing the weighted-mean estimating equation together with the logistic score equations. With independent samples of fixed source sizes, let $s_H^2$ and $s_A^2$ be the within-source sample variances of these contributions, using divisors $n_H-1$ and $m-1$. The simulation and real-data implementations use
\begin{equation}
\label{eq:joint-variance}
\hat\tau_{0,\mathrm{joint}}^2(F;m)
=\frac{n_Hs_H^2+ms_A^2}{W^2}.
\end{equation}
The contribution from the target sample reflects uncertainty in the fitted participation score. Untruncated fitted probabilities are used in the logistic information and score contributions; truncation changes the weighted-mean derivative through $d_j$.

The unadjusted variance estimate was the historical sample variance divided by $n_H$. For standardized designs the final test used estimated variance $1/N+\hat\tau_{0,\mathrm{joint}}^2(F;m)$. At planning and for candidate and oracle calculations, $m=1{,}000$; at each SSR review, $m$ is the current number of active covariate records. The current-$m$ estimate is inserted into Equation~\eqref{eq-NF} without extrapolating it to a future target-sample size. SSR retains the historical mean and joint variance from the last review, including when recruitment stops between reviews. The standardized fixed designs instead refit at their final active-arm sample sizes, where $m=N$. The joint estimator is used throughout these simulations and in the real-data design calculations described in Web Appendix C.
For the conservative rule in the one-covariate simulations, let $\hat\mu_n$ and $s_n/\sqrt n$ denote the active covariate mean and its estimated standard error. The numerical derivative of log required size is evaluated at $\hat\mu_n\pm h$, with $h=\max(0.05,s_n/(2\sqrt n))$.

\clearpage
\subsection{Sensitivity analysis with 150 historical controls}

Tables~\ref{tab:sim-linear-n150}--\ref{tab:sim-threshold-n150} report the sensitivity analysis with $n_H=150$, assessing performance with more limited historical information. The fixed design contrast of 0.5, model settings, review schedule, and 10,000 replicates per setting were the same as in the primary analysis. Candidate requirements were restricted to $[20,150]$ and the absolute SSR cap was 150 rather than 500. Comparisons therefore reflect both the smaller historical sample and the changed numerical cap.

Table~\ref{tab:sim-linear-n150} reports the linear-model sensitivity results with 150 historical controls. At $\mu_A=1$, the population-standardized fixed design had power 0.678 and mean enrollment 40.4. Point and conservative SSR increased mean enrollment to 73.8 and 85.2, respectively, with power 0.792 and 0.811. The oracle had power 0.776 and mean enrollment 72.8. Corresponding type I error rates were 0.0558 for point SSR, 0.0574 for conservative SSR, and 0.0537 for the oracle. Allowing SSR to exceed the maximum candidate requirement improved flexibility; conservative SSR attained an empirical power slightly above 0.80 in this setting.

The quadratic and threshold sensitivity results with 150 historical controls are shown in Tables~\ref{tab:sim-quadratic-n150} and~\ref{tab:sim-threshold-n150}. At $\mu_A=1$, conservative SSR had mean enrollment 89.5 and 90.4 and power 0.801 and 0.803, respectively. Corresponding oracle powers were 0.762 and 0.755. Null rejection probabilities for point and conservative SSR were 0.0546 and 0.0564 in the quadratic model and 0.0570 and 0.0549 in the threshold model. Thus, the higher enrollment supported target power for the conservative rule in these settings, while mild null inflation remained. 

The joint variance did not eliminate finite-sample type I error inflation with 150 historical controls. Across the standardized designs, null rejection probabilities ranged from 0.0472 to 0.0580; at $\mu_A=1$, the two SSR rules ranged from 0.0546 to 0.0574 across the three models. For threshold-model point SSR at $\mu_A=1$, the empirical variance of the test numerator divided by the average reported squared standard error was 1.039. Analytically integrating over the independent normal active-arm mean gave a null rejection estimate of 0.0562 (Monte Carlo SE 0.0011), compared with 0.0554 (SE 0.0010) for the oracle. This diagnostic removes active-outcome simulation noise under the stated generator and supports a remaining calibration discrepancy that is not confined to SSR. Joint score and target estimation alone therefore does not establish control after adaptive stopping.

\input{simulation_table_linear_n150}
\input{simulation_table_quadratic_n150}
\input{simulation_table_threshold_n150}

\end{document}

%% file: def.tex
\def\x{{\text{\boldmath $x$}}}

\def\X{{\text{\boldmath $X$}}}

%% file: real_data_design_table.tex
\begin{table}[htbp]
\centering
\caption{Sample size determination and re-estimation in the real-data example.}
\label{tab:real_data_design}
\begin{tabular}{lp{0.35\linewidth}lcccc}\toprule
Case & Method & Review $n$ & $|\Delta_{\mathrm{design}}|$ & $\widehat{\mathrm{Var}}(\hat\eta_0)$ & $N_{\mathrm{target}}$ & $N_{\mathrm{final}}$ \\
\midrule
Case 1 & Proposed SSR, point estimate & 20 & 6.00 & 1.1952 & 50 &  \\
Case 1 & Proposed SSR, point estimate & 30 & 6.00 & 1.2192 & 51 &  \\
Case 1 & Proposed SSR, point estimate & 40 & 6.00 & 1.2489 & 51 & 51 \\
Case 1 & Proposed SSR, conservative & 20 & 6.00 & 1.1952 & 50 &  \\
Case 1 & Proposed SSR, conservative & 30 & 6.00 & 1.2192 & 51 &  \\
Case 1 & Proposed SSR, conservative & 40 & 6.00 & 1.2489 & 51 & 51 \\
Case 1 & Population-standardized fixed & Planning & 6.00 & 1.1788 & 50 & 50 \\
Case 1 & Unadjusted fixed & Planning & 6.00 & 1.1352 & 49 & 49 \\
\midrule
Case 2 & Proposed SSR, point estimate & 20 & 6.00 & 3.8648 & 232 &  \\
Case 2 & Proposed SSR, point estimate & 40 & 6.00 & 2.5580 & 84 &  \\
Case 2 & Proposed SSR, point estimate & 60 & 6.00 & 1.9430 & 64 & 64 \\
Case 2 & Proposed SSR, conservative & 20 & 6.00 & 3.8648 & 232 &  \\
Case 2 & Proposed SSR, conservative & 40 & 6.00 & 2.5580 & 95 &  \\
Case 2 & Proposed SSR, conservative & 60 & 6.00 & 1.9430 & 71 & 71 \\
Case 2 & Population-standardized fixed & Planning & 6.00 & 1.1057 & 49 & 49 \\
Case 2 & Unadjusted fixed & Planning & 6.00 & 1.1352 & 49 & 49 \\
\bottomrule
\end{tabular}
\end{table}

%% file: simulation_table_linear.tex
\begin{table}[!htbp]
\centering
\caption{Operating characteristics for the linear outcome model with $n_H=500$ historical controls.}
\label{tab:sim-linear}
\begin{tabular}{lllcc}
\hline
$\mu_A$ & Method & $N_{\mathrm{final}}$, mean (SD) & Type I error & Power
\\
\hline
0.0 & Oracle & 34.0 (0.0) & 0.047 & 0.800
\\
0.0 & \textbf{Proposed SSR, point} & 34.1 (0.3) & 0.046 & 0.799
\\
0.0 & \textbf{Proposed SSR, conservative} & 34.2 (0.7) & 0.046 & 0.802
\\
0.0 & Population-standardized fixed & 34.0 (0.0) & 0.047 & 0.800
\\
0.0 & Unadjusted fixed & 34.0 (0.0) & 0.047 & 0.802
\\
0.0 & Max-scenario fixed & 38.9 (15.6) & 0.048 & 0.841
\\
\hline
0.5 & Oracle & 34.7 (0.5) & 0.049 & 0.800
\\
0.5 & \textbf{Proposed SSR, point} & 35.5 (1.1) & 0.048 & 0.807
\\
0.5 & \textbf{Proposed SSR, conservative} & 36.7 (1.9) & 0.048 & 0.821
\\
0.5 & Population-standardized fixed & 34.0 (0.0) & 0.047 & 0.793
\\
0.5 & Unadjusted fixed & 34.0 (0.0) & 0.061 & 0.888
\\
0.5 & Max-scenario fixed & 39.1 (16.5) & 0.046 & 0.832
\\
\hline
1.0 & Oracle & 39.5 (18.8) & 0.053 & 0.787
\\
1.0 & \textbf{Proposed SSR, point} & 41.3 (15.6) & 0.051 & 0.810
\\
1.0 & \textbf{Proposed SSR, conservative} & 45.5 (22.0) & 0.051 & 0.832
\\
1.0 & Population-standardized fixed & 34.0 (0.0) & 0.054 & 0.747
\\
1.0 & Unadjusted fixed & 34.0 (0.1) & 0.206 & 0.978
\\
1.0 & Max-scenario fixed & 39.4 (18.9) & 0.053 & 0.787
\\
\hline
\end{tabular}
\end{table}

%% file: simulation_table_quadratic.tex
\begin{table}[!htbp]
\centering
\caption{Operating characteristics for the quadratic outcome model with $n_H=500$ historical controls.}
\label{tab:sim-quadratic}
\begin{tabular}{lllcc}
\hline
$\mu_A$ & Method & $N_{\mathrm{final}}$, mean (SD) & Type I error & Power
\\
\hline
0.0 & Oracle & 34.0 (0.0) & 0.052 & 0.801
\\
0.0 & \textbf{Proposed SSR, point} & 34.1 (0.3) & 0.053 & 0.802
\\
0.0 & \textbf{Proposed SSR, conservative} & 34.3 (0.8) & 0.053 & 0.803
\\
0.0 & Population-standardized fixed & 34.0 (0.0) & 0.052 & 0.801
\\
0.0 & Unadjusted fixed & 34.0 (0.0) & 0.053 & 0.802
\\
0.0 & Max-scenario fixed & 38.9 (15.9) & 0.051 & 0.843
\\
\hline
0.5 & Oracle & 34.7 (0.5) & 0.046 & 0.790
\\
0.5 & \textbf{Proposed SSR, point} & 36.3 (1.4) & 0.048 & 0.806
\\
0.5 & \textbf{Proposed SSR, conservative} & 37.5 (2.0) & 0.047 & 0.820
\\
0.5 & Population-standardized fixed & 34.0 (0.0) & 0.048 & 0.782
\\
0.5 & Unadjusted fixed & 34.0 (0.1) & 0.079 & 0.915
\\
0.5 & Max-scenario fixed & 39.6 (18.7) & 0.046 & 0.827
\\
\hline
1.0 & Oracle & 40.4 (21.2) & 0.052 & 0.766
\\
1.0 & \textbf{Proposed SSR, point} & 44.2 (16.5) & 0.051 & 0.809
\\
1.0 & \textbf{Proposed SSR, conservative} & 49.0 (26.4) & 0.049 & 0.835
\\
1.0 & Population-standardized fixed & 34.0 (0.1) & 0.051 & 0.721
\\
1.0 & Unadjusted fixed & 34.0 (0.2) & 0.409 & 0.994
\\
1.0 & Max-scenario fixed & 40.4 (20.9) & 0.053 & 0.766
\\
\hline
\end{tabular}
\end{table}

%% file: simulation_table_threshold.tex
\begin{table}[!htbp]
\centering
\caption{Operating characteristics for the threshold outcome model with $n_H=500$ historical controls.}
\label{tab:sim-threshold}
\begin{tabular}{lllcc}
\hline
$\mu_A$ & Method & $N_{\mathrm{final}}$, mean (SD) & Type I error & Power
\\
\hline
0.0 & Oracle & 34.0 (0.0) & 0.047 & 0.807
\\
0.0 & \textbf{Proposed SSR, point} & 34.1 (0.3) & 0.048 & 0.807
\\
0.0 & \textbf{Proposed SSR, conservative} & 34.2 (0.7) & 0.048 & 0.808
\\
0.0 & Population-standardized fixed & 34.0 (0.0) & 0.047 & 0.807
\\
0.0 & Unadjusted fixed & 34.0 (0.0) & 0.049 & 0.808
\\
0.0 & Max-scenario fixed & 38.8 (13.5) & 0.050 & 0.845
\\
\hline
0.5 & Oracle & 35.1 (0.4) & 0.053 & 0.759
\\
0.5 & \textbf{Proposed SSR, point} & 40.7 (2.0) & 0.052 & 0.810
\\
0.5 & \textbf{Proposed SSR, conservative} & 41.5 (2.5) & 0.053 & 0.819
\\
0.5 & Population-standardized fixed & 34.2 (0.4) & 0.052 & 0.753
\\
0.5 & Unadjusted fixed & 34.6 (0.5) & 0.207 & 0.974
\\
0.5 & Max-scenario fixed & 39.3 (16.3) & 0.053 & 0.793
\\
\hline
1.0 & Oracle & 39.6 (16.2) & 0.051 & 0.766
\\
1.0 & \textbf{Proposed SSR, point} & 45.3 (15.2) & 0.052 & 0.816
\\
1.0 & \textbf{Proposed SSR, conservative} & 49.1 (20.5) & 0.050 & 0.836
\\
1.0 & Population-standardized fixed & 34.2 (0.4) & 0.051 & 0.728
\\
1.0 & Unadjusted fixed & 34.5 (0.5) & 0.572 & 0.999
\\
1.0 & Max-scenario fixed & 39.5 (15.9) & 0.050 & 0.765
\\
\hline
\end{tabular}
\end{table}

%% file: simulation_table_linear_n150.tex
\begin{table}[!htbp]
\centering
\caption{Operating characteristics for the linear outcome model with $n_H=150$ historical controls.}
\label{tab:sim-linear-n150}
\begin{tabular}{lllcc}
\hline
$\mu_A$ & Method & $N_{\mathrm{final}}$, mean (SD) & Type I error & Power
\\
\hline
0.0 & Oracle & 40.2 (1.3) & 0.055 & 0.792
\\
0.0 & \textbf{Proposed SSR, point} & 40.9 (1.6) & 0.054 & 0.795
\\
0.0 & \textbf{Proposed SSR, conservative} & 41.9 (2.4) & 0.055 & 0.804
\\
0.0 & Population-standardized fixed & 40.3 (1.3) & 0.055 & 0.792
\\
0.0 & Unadjusted fixed & 40.2 (1.3) & 0.055 & 0.795
\\
0.0 & Max-scenario fixed & 69.3 (32.4) & 0.047 & 0.894
\\
\hline
0.5 & Oracle & 43.9 (5.0) & 0.049 & 0.791
\\
0.5 & \textbf{Proposed SSR, point} & 45.5 (7.1) & 0.050 & 0.807
\\
0.5 & \textbf{Proposed SSR, conservative} & 49.4 (10.7) & 0.051 & 0.823
\\
0.5 & Population-standardized fixed & 40.3 (1.3) & 0.051 & 0.769
\\
0.5 & Unadjusted fixed & 40.4 (1.3) & 0.063 & 0.881
\\
0.5 & Max-scenario fixed & 71.0 (33.4) & 0.049 & 0.879
\\
\hline
1.0 & Oracle & 72.8 (34.1) & 0.054 & 0.776
\\
1.0 & \textbf{Proposed SSR, point} & 73.8 (33.5) & 0.056 & 0.792
\\
1.0 & \textbf{Proposed SSR, conservative} & 85.2 (36.6) & 0.057 & 0.811
\\
1.0 & Population-standardized fixed & 40.4 (1.3) & 0.058 & 0.678
\\
1.0 & Unadjusted fixed & 40.7 (1.3) & 0.205 & 0.977
\\
1.0 & Max-scenario fixed & 72.9 (34.2) & 0.054 & 0.774
\\
\hline
\end{tabular}
\end{table}

%% file: simulation_table_quadratic_n150.tex
\begin{table}[!htbp]
\centering
\caption{Operating characteristics for the quadratic outcome model with $n_H=150$ historical controls.}
\label{tab:sim-quadratic-n150}
\begin{tabular}{lllcc}
\hline
$\mu_A$ & Method & $N_{\mathrm{final}}$, mean (SD) & Type I error & Power
\\
\hline
0.0 & Oracle & 40.3 (1.3) & 0.050 & 0.803
\\
0.0 & \textbf{Proposed SSR, point} & 40.9 (1.6) & 0.051 & 0.805
\\
0.0 & \textbf{Proposed SSR, conservative} & 41.9 (2.4) & 0.051 & 0.812
\\
0.0 & Population-standardized fixed & 40.3 (1.3) & 0.050 & 0.803
\\
0.0 & Unadjusted fixed & 40.2 (1.3) & 0.050 & 0.804
\\
0.0 & Max-scenario fixed & 70.3 (33.1) & 0.051 & 0.899
\\
\hline
0.5 & Oracle & 43.9 (5.0) & 0.052 & 0.786
\\
0.5 & \textbf{Proposed SSR, point} & 46.3 (7.5) & 0.051 & 0.805
\\
0.5 & \textbf{Proposed SSR, conservative} & 50.2 (10.8) & 0.053 & 0.826
\\
0.5 & Population-standardized fixed & 40.3 (1.3) & 0.050 & 0.766
\\
0.5 & Unadjusted fixed & 40.6 (1.3) & 0.079 & 0.911
\\
0.5 & Max-scenario fixed & 71.7 (34.0) & 0.055 & 0.876
\\
\hline
1.0 & Oracle & 76.2 (36.1) & 0.054 & 0.762
\\
1.0 & \textbf{Proposed SSR, point} & 78.7 (34.9) & 0.055 & 0.788
\\
1.0 & \textbf{Proposed SSR, conservative} & 89.5 (37.3) & 0.056 & 0.801
\\
1.0 & Population-standardized fixed & 40.4 (1.3) & 0.058 & 0.666
\\
1.0 & Unadjusted fixed & 41.1 (1.4) & 0.400 & 0.995
\\
1.0 & Max-scenario fixed & 76.0 (35.9) & 0.056 & 0.765
\\
\hline
\end{tabular}
\end{table}

%% file: simulation_table_threshold_n150.tex
\begin{table}[!htbp]
\centering
\caption{Operating characteristics for the threshold outcome model with $n_H=150$ historical controls.}
\label{tab:sim-threshold-n150}
\begin{tabular}{lllcc}
\hline
$\mu_A$ & Method & $N_{\mathrm{final}}$, mean (SD) & Type I error & Power
\\
\hline
0.0 & Oracle & 40.3 (1.3) & 0.053 & 0.800
\\
0.0 & \textbf{Proposed SSR, point} & 40.9 (1.6) & 0.052 & 0.804
\\
0.0 & \textbf{Proposed SSR, conservative} & 42.0 (2.4) & 0.052 & 0.813
\\
0.0 & Population-standardized fixed & 40.3 (1.3) & 0.053 & 0.800
\\
0.0 & Unadjusted fixed & 40.2 (1.3) & 0.052 & 0.803
\\
0.0 & Max-scenario fixed & 70.4 (33.2) & 0.052 & 0.903
\\
\hline
0.5 & Oracle & 45.4 (5.1) & 0.049 & 0.756
\\
0.5 & \textbf{Proposed SSR, point} & 53.2 (8.0) & 0.047 & 0.803
\\
0.5 & \textbf{Proposed SSR, conservative} & 56.7 (11.1) & 0.048 & 0.826
\\
0.5 & Population-standardized fixed & 41.5 (1.5) & 0.050 & 0.732
\\
0.5 & Unadjusted fixed & 43.0 (1.8) & 0.200 & 0.975
\\
0.5 & Max-scenario fixed & 72.4 (33.3) & 0.048 & 0.845
\\
\hline
1.0 & Oracle & 74.8 (34.1) & 0.054 & 0.755
\\
1.0 & \textbf{Proposed SSR, point} & 80.7 (33.3) & 0.057 & 0.786
\\
1.0 & \textbf{Proposed SSR, conservative} & 90.4 (35.5) & 0.055 & 0.803
\\
1.0 & Population-standardized fixed & 41.5 (1.5) & 0.055 & 0.663
\\
1.0 & Unadjusted fixed & 42.6 (1.7) & 0.578 & 0.998
\\
1.0 & Max-scenario fixed & 74.6 (34.0) & 0.054 & 0.753
\\
\hline
\end{tabular}
\end{table}

%% file: main.bib
@article{makuch1980sample,
  title={Sample size considerations for non-randomized comparative studies},
  author={Makuch, Robert W and Simon, Richard M},
  journal={Journal of Chronic Diseases},
  volume={33},
  number={3},
  pages={175--181},
  year={1980},
  doi={10.1016/0021-9681(80)90017-X},
  publisher={Elsevier}
}

@article{zhang2010calculating,
  title={Calculating sample size in trials using historical controls},
  author={Zhang, Song and Cao, Jing and Ahn, Chul},
  journal={Clinical Trials},
  volume={7},
  number={4},
  pages={343--353},
  year={2010},
  doi={10.1177/1740774510373629},
  publisher={SAGE Publications Sage UK: London, England}
}

@article{loiseau2022external,
  title={External control arm analysis: an evaluation of propensity score approaches, G-computation, and doubly debiased machine learning},
  author={Loiseau, Nicolas and Trichelair, Paul and He, Maxime and Andreux, Mathieu and Zaslavskiy, Mikhail and Wainrib, Gilles and others},
  journal={BMC Medical Research Methodology},
  volume={22},
  number={1},
  pages={335},
  year={2022},
  doi={10.1186/s12874-022-01799-z},
  publisher={Springer}
}

@article{schoenfeld2019design,
  title={Design and analysis of a clinical trial using previous trials as historical control},
  author={Schoenfeld, David Alan and Finkelstein, Dianne M and Macklin, Eric and Zach, Neta and Ennist, David L and Taylor, Albert A and others},
  journal={Clinical Trials},
  volume={16},
  number={5},
  pages={531--538},
  year={2019},
  doi={10.1177/1740774519858914},
  publisher={SAGE Publications Sage UK: London, England}
}

@article{o2002sample,
  title={Sample size calculation for a historically controlled clinical trial with adjustment for covariates},
  author={O'Malley, A James and Normand, Sharon-Lise T and Kuntz, Richard E},
  journal={Journal of Biopharmaceutical Statistics},
  volume={12},
  number={2},
  pages={227--247},
  year={2002},
  doi={10.1081/BIP-120015745},
  publisher={Taylor \& Francis}
}

@article{pocock1976historical,
  title={The combination of randomized and historical controls in clinical trials},
  author={Pocock, Stuart J},
  journal={Journal of Chronic Diseases},
  volume={29},
  number={3},
  pages={175--188},
  year={1976},
  doi={10.1016/0021-9681(76)90044-8},
  publisher={Elsevier}
}

@article{viele2014historical,
  title={Use of historical control data for assessing treatment effects in clinical trials},
  author={Viele, Kert and Berry, Scott and Neuenschwander, Beat and Amzal, Billy and Chen, Fang and Enas, Nathan and others},
  journal={Pharmaceutical Statistics},
  volume={13},
  number={1},
  pages={41--54},
  year={2014},
  doi={10.1002/pst.1589},
  publisher={Wiley}
}

@article{thorlund2020synthetic,
  title={Synthetic and external controls in clinical trials: A primer for researchers},
  author={Thorlund, Kristian and Dron, Louis and Park, Jay JH and Mills, Edward J},
  journal={Clinical Epidemiology},
  volume={12},
  pages={457--467},
  year={2020},
  doi={10.2147/CLEP.S242097},
  publisher={Dove Medical Press}
}

@article{seeger2020external,
  title={Methods for external control groups for single arm trials or long-term uncontrolled extensions to randomized clinical trials},
  author={Seeger, John D and Davis, Kourtney J and Iannacone, Michelle R and Zhou, Wei and Dreyer, Nancy and Winterstein, Almut G and others},
  journal={Pharmacoepidemiology and Drug Safety},
  volume={29},
  number={11},
  pages={1382--1392},
  year={2020},
  doi={10.1002/pds.5141},
  publisher={Wiley}
}

@article{burcu2020rwe,
  title={Real-world evidence to support regulatory decision-making for medicines: Considerations for external control arms},
  author={Burcu, Mehmet and Dreyer, Nancy A and Franklin, Jessica M and Blum, Michael D and Critchlow, Cathy W and Perfetto, Eleanor M and others},
  journal={Pharmacoepidemiology and Drug Safety},
  volume={29},
  number={10},
  pages={1228--1235},
  year={2020},
  doi={10.1002/pds.4975},
  publisher={Wiley}
}

@article{mishraKalyani2022external,
  title={External control arms in oncology: current use and future directions},
  author={Mishra-Kalyani, Pallavi S and Amiri Kordestani, Laleh and Rivera, Donna R and Singh, Harpreet and Ibrahim, Amna and DeClaro, Richard A and others},
  journal={Annals of Oncology},
  volume={33},
  number={4},
  pages={376--383},
  year={2022},
  doi={10.1016/j.annonc.2021.12.015},
  publisher={Elsevier}
}

@article{rosenbaumRubin1983,
  title={The central role of the propensity score in observational studies for causal effects},
  author={Rosenbaum, Paul R and Rubin, Donald B},
  journal={Biometrika},
  volume={70},
  number={1},
  pages={41--55},
  year={1983},
  doi={10.1093/biomet/70.1.41},
  publisher={Oxford University Press}
}

@article{dahabreh2020transport,
  title={Extending inferences from a randomized trial to a new target population},
  author={Dahabreh, Issa J and Robertson, Sarah E and Steingrimsson, Jon A and Stuart, Elizabeth A and Hernan, Miguel A},
  journal={Statistics in Medicine},
  volume={39},
  number={14},
  pages={1999--2014},
  year={2020},
  doi={10.1002/sim.8426},
  publisher={Wiley}
}

@article{gouldShih1992,
  title={Sample size re-estimation without unblinding for normally distributed outcomes with unknown variance},
  author={Gould, A Lawrence and Shih, Weichung Joseph},
  journal={Communications in Statistics - Theory and Methods},
  volume={21},
  number={10},
  pages={2833--2853},
  year={1992},
  doi={10.1080/03610929208830947},
  publisher={Taylor \& Francis}
}

@article{gould2001ssr,
  title={Sample size re-estimation: Recent developments and practical considerations},
  author={Gould, A Lawrence},
  journal={Statistics in Medicine},
  volume={20},
  number={17-18},
  pages={2625--2643},
  year={2001},
  doi={10.1002/sim.733},
  publisher={Wiley}
}

@article{proschanHunsberger1995,
  title={Designed extension of studies based on conditional power},
  author={Proschan, Michael A and Hunsberger, Sally A},
  journal={Biometrics},
  volume={51},
  number={4},
  pages={1315--1324},
  year={1995},
  doi={10.2307/2533262},
  publisher={International Biometric Society}
}

@article{kojima2026hybrid,
  title={Sample size re-estimation in blinded hybrid-control design using inverse probability weighting},
  author={Kojima, Masahiro and Orihara, Shunichiro and Hanada, Keisuke and Ohigashi, Tomohiro},
  journal={Statistics in Medicine},
  volume={45},
  number={3-5},
  pages={e70429},
  year={2026},
  doi={10.1002/sim.70429},
  publisher={Wiley}
}

@article{xuFriede2026monitoring,
  title={A note on blinded continuous monitoring for continuous outcomes},
  author={Xu, Long-Hao and Friede, Tim},
  journal={Statistics \& Probability Letters},
  volume={228},
  pages={110575},
  year={2026},
  doi={10.1016/j.spl.2025.110575},
  publisher={Elsevier}
}

@misc{maeda2026blinded,
  title={Blinded sample size re-estimation accounting for uncertainty in mid-trial estimation},
  author={Maeda, Hirotada and Hattori, Satoshi and Friede, Tim},
  year={2026},
  eprint={2602.03218},
  archivePrefix={arXiv},
  primaryClass={stat.ME},
  doi={10.48550/arXiv.2602.03218}
}

@misc{ichE9R1,
  title={{ICH E9(R1)} Statistical Principles for Clinical Trials: Addendum: Estimands and Sensitivity Analysis in Clinical Trials},
  author={{ICH}},
  year={2019},
  note={Step 5 guideline. Available at \url{https://www.pmda.go.jp/files/000269157.pdf}. Accessed July 22, 2026}
}

@article{stuart2011generalizability,
  title={The use of propensity scores to assess the generalizability of results from randomized trials},
  author={Stuart, Elizabeth A and Cole, Stephen R and Bradshaw, Catherine P and Leaf, Philip J},
  journal={Journal of the Royal Statistical Society: Series A (Statistics in Society)},
  volume={174},
  number={2},
  pages={369--386},
  year={2011},
  doi={10.1111/j.1467-985X.2010.00673.x},
  publisher={Wiley}
}

@article{coleHernan2008weights,
  title={Constructing inverse probability weights for marginal structural models},
  author={Cole, Stephen R and Hernan, Miguel A},
  journal={American Journal of Epidemiology},
  volume={168},
  number={6},
  pages={656--664},
  year={2008},
  doi={10.1093/aje/kwn164},
  publisher={Oxford University Press}
}

@article{austinStuart2015iptw,
  title={Moving towards best practice when using inverse probability of treatment weighting {(IPTW)} using the propensity score to estimate causal treatment effects in observational studies},
  author={Austin, Peter C and Stuart, Elizabeth A},
  journal={Statistics in Medicine},
  volume={34},
  number={28},
  pages={3661--3679},
  year={2015},
  doi={10.1002/sim.6607},
  publisher={Wiley}
}

@book{chow2017sample,
  title={Sample Size Calculations in Clinical Research},
  author={Chow, Shein-Chung and Shao, Jun and Wang, Hansheng and Lokhnygina, Yuliya},
  edition={3},
  year={2017},
  publisher={Chapman and Hall/CRC},
  address={Boca Raton},
  doi={10.1201/9781315183084}
}

@article{friedeKieser2006internal,
  title={Sample size recalculation in internal pilot study designs: A review},
  author={Friede, Tim and Kieser, Meinhard},
  journal={Biometrical Journal},
  volume={48},
  number={4},
  pages={537--555},
  year={2006},
  doi={10.1002/bimj.200510238},
  publisher={Wiley}
}

@article{aisen2008bvitamin,
  title={High-dose {B} vitamin supplementation and cognitive decline in {Alzheimer} disease: A randomized controlled trial},
  author={Aisen, Paul S and Schneider, Lon S and Sano, Mary and Diaz-Arrastia, Ramon and van Dyck, Christopher H and Weiner, Myron F and others},
  journal={JAMA},
  volume={300},
  number={15},
  pages={1774--1783},
  year={2008},
  doi={10.1001/jama.300.15.1774}
}

@article{quinn2010dha,
  title={Docosahexaenoic acid supplementation and cognitive decline in {Alzheimer} disease: A randomized trial},
  author={Quinn, Joseph F and Raman, Rema and Thomas, Ronald G and Yurko-Mauro, Karin and Nelson, Edward B and Van Dyck, Christopher and others},
  journal={JAMA},
  volume={304},
  number={17},
  pages={1903--1911},
  year={2010},
  doi={10.1001/jama.2010.1510}
}

@article{qi2022historical,
  title={Incorporating historical controls in clinical trials with longitudinal outcomes using the modified power prior},
  author={Qi, Hongchao and Rizopoulos, Dimitris and Lesaffre, Emmanuel and van Rosmalen, Joost},
  journal={Pharmaceutical Statistics},
  volume={21},
  number={5},
  pages={818--834},
  year={2022},
  doi={10.1002/pst.2195},
  publisher={Wiley}
}

@article{hanada2026inmass,
  title={Integrating Meta-Analysis Into a Specific Study ({InMASS}) for Estimating the Target-Population Average Treatment Effect},
  author={Hanada, Keisuke and Kojima, Masahiro},
  journal={Biometrical Journal},
  volume={68},
  number={5},
  pages={e70169},
  year={2026},
  doi={10.1002/bimj.70169}
}

@techreport{fda2023rare,
  author      = {{U.S. Food and Drug Administration}},
  title       = {{Rare Diseases: Considerations for the Development of Drugs and Biological Products}},
  institution = {U.S. Food and Drug Administration},
  type        = {Guidance for Industry},
  year        = {2023},
  month       = dec,
  url         = {https://www.fda.gov/media/119757/download}
}
